\documentclass[12pt]{amsart}

\usepackage{amsmath,amsthm,amssymb,amsfonts,verbatim,graphics}
\usepackage[colorlinks,
            linkcolor=blue,
            anchorcolor=blue,
            citecolor=blue
            ]{hyperref}
\usepackage[hmargin=1.2in,vmargin=1.2in]{geometry}
\usepackage{latexsym}
\usepackage{graphicx}
\usepackage{float}

\usepackage{array}

\title[Locally Repairable Codes]{Good Locally Repairable Codes via Propagation Rules}
\author{Shu Liu}\address{Natl Key Lab Sci and Technol Commun, University of Electronic Science and Technology of China, Chengdu, China.} \email{shuliu@uestc.edu.cn}
\author{Liming Ma}\address{School of Mathematical Sciences, University of Science and Technology of China, Hefei 230026, China
}\email{lmma20@ustc.edu.cn}
\author{Tingyi Wu}\address{Theory Lab, Central Research Institute, 2012 Labs, Huawei Technology Co. Ltd.
}\email{wu.ting.yi@huawei.com}
\author{Chaoping Xing} \address{School of Electronics, Information and Electric Engineering, Shanghai Jiao Tong University,
China 200240}\email{xingcp@sjtu.edu.cn}
\date{}

\newtheorem{lemma}{Lemma}[section]
\newtheorem{theorem}[lemma]{Theorem}
\newtheorem{cor}[lemma]{Corollary}

\newtheorem{prop}[lemma]{Proposition}
\newtheorem{ex}[lemma]{Example}
\newtheorem{rem}[lemma]{Remark}
\newtheorem{defn}{Definition}

\theoremstyle{remark}
\newtheorem{rmk}{Remark}

\renewcommand{\epsilon}{\varepsilon}
\renewcommand{\le}{\leqslant}
\renewcommand{\ge}{\geqslant}

\newcommand{\vnote}[1]{}

\def\PP{\mathbb{P}}

\def\F{\mathbb{F}}

\def \mC {\mathcal{C}}
\def \mA {\mathcal{A}}

\def \mA {\mathcal{A}}

\def \mC {\mathcal{C}}

\def \mL {\mathcal{L}}

\def \mP {\mathcal{P}}

\def \Xi {{X^{[i]}}}

\newcommand{\Ga}{\alpha}

\newcommand{\Gd}{\delta}

\newcommand{\Gl}{\lambda}

\def \bc {{\bf c}}

\def \bh {{\bf h}}
\def \bx {{\bf x}}

\def \bu {{\bf u}}

\def \bo {{\bf 0}}

\def\supp {{\rm supp }}

\def\LRC {{\rm locally repairable code\ }}
\def\LRCs {{\rm locally repairable codes\ }}

\begin{document}
\maketitle

\begin{abstract} 
In classical coding theory, it is common to construct new codes via propagation rules. There are various propagation rules to construct classical block codes. However, propagation rules have not been extensively explored for constructions of locally repairable codes. In this paper, we introduce a few propagation rules to construct good locally repairable codes. To our surprise, these simple propagation rules produce a few interesting results. Firstly, by concatenating a locally repairable code as an inner code with a classical block code as an outer code, we obtain quite a few dimension-optimal binary locally repairable codes. Secondly, from this concatenation, we explicitly build a family of locally repairable codes that exceeds the Zyablov-type bound. Thirdly, by a lengthening propagation rule that adds some rows and columns from a parity-check matrix of a given linear code, we are able to produce a family of  dimension-optimal binary \LRCs from the extended Hamming codes, and to convert a classical maximum distance separable (MDS) code into a Singleton-optimal locally repairable code. Furthermore, via the lengthening propagation rule, we greatly simplify the construction of a family of locally repairable codes in \cite[Theorem 5]{MX20} that breaks the asymptotic Gilbert-Varshamov bound. In addition, we make use of  three other propagation rules to produce more dimension-optimal binary locally repairable codes. Finally, one of phenomena that we observe in this paper is that some trivial propagation rules in classical block codes do not hold anymore for locally repairable codes.
\end{abstract}

\section{Introduction}
Since the first work on locally repairable codes in \cite{HL07,HCL07}, construction of good locally repairable codes has been a central problem in the study of locally repairable codes \cite{BTV17,BHHMV17,CMST20,FY14,GHSY12,GXY19,HXSCFY20,J19,LMX19,LXY19,PD14,SES19,SRKV13,TB14,TBF16,TPD16}. 
A locally repairable code is just a block code with an additional parameter called locality. 
It was shown in \cite{GHSY12} that the minimum distance $d(C)$ of an $[n,k,d]$-linear code $C$ with locality $r$ is upper bounded by
 \begin{equation}\label{eq:x1}
 d(C)\leq n-k-\left\lceil \frac kr\right\rceil+2.
 \end{equation}
The bound \eqref{eq:x1} is called the Singleton-type bound for locally repairable codes, and hence any code achieving this bound is called an optimal locally repairable code or a Singleton-optimal locally repairable code.

\subsection{Known Results}
The construction of optimal  locally repairable codes is of both theoretical interest and practical importance. A class of codes constructed earlier and known as pyramid codes \cite{HCL07} are shown to be optimal locally repairable codes.  In \cite{SRKV13},  Silberstein {\it et al.}  proposed a two-level construction based on the Gabidulin codes combined with a single parity-check $(r+1,r)$ code. Another construction \cite{TPD16} used two layers of MDS codes, a Reed-Solomon code and a special $(r+1,r)$ MDS code. A common shortcoming of these constructions is that the size of the code alphabet is an exponential function of the code length, and complicating the implementation. There  was an earlier construction of optimal locally repairable codes given in \cite{PKLK12} with  alphabet  size comparable to code length. However, the rate of the code is very close to $1$.

A remarkable construction of optimal locally repairable codes via subcodes of Reed-Solomon codes was given by Tamo and Barg in \cite{TB14}. There are some constraints on choice of locality due to the existence of good polynomials, and the length is upper bounded by the code alphabet size $q$. This construction was generalized via the automorphism group of a rational function field in \cite{JMX20} and it turns out that there are more flexibility on locality and the code length can be $q+1$. Based on the classical MDS conjecture, one should wonder if $q$-ary optimal locally repairable codes can have length bigger than $q+1$. Surprisingly, several optimal locally repairable codes with length around $q^2$ are constructed from algebraic surfaces in \cite{BHHMV17}. By making use of automorphism groups of elliptic curves, optimal locally repairable codes with length up to $q+2\sqrt{q}$ can be constructed \cite{LMX19,MX20}. Algebraic surfaces are employed to construct two families of optimal locally repairable codes \cite{SVV21}. In these constructions, there are no restriction on the minimum distance of optimal locally repairable codes.

As for optimal locally repairable codes with small distances, the length of optimal locally repairable codes could be much larger than the alphabet size $q$. 
Arbitrary long optimal locally repairable codes can be constructed by cyclic codes for distance $d=3, 4$ \cite{LXY19}, and the length of an optimal $q$-ary locally repairable code is at most roughly $O(dq^3)$ for distance $d\geq 5$ \cite{GXY19}. Jin provided an explicit construction of $q$-ary optimal locally repairable codes with length $O(q^2)$ via binary constant-weight codes for distance $5$ and $6$ \cite{J19}. From extreme graph theory, there exists explicit construction of optimal locally repairable codes with super-linear length for distance $d\geq 7$ \cite{XY19}. 
Over the past few years, locally repairable codes have been generalized to correct multiple erasure errors \cite{CMST20, CFXF19, CXHF18, FF20} and correct erasures from multiple recovery sets \cite{CMST20-4, JKZ20,SES19}.

The Singleton-type bound of locally repairable codes is not always tight. In \cite{CM15}, Cadambe and Mazumdar derived a field-dependent bound, known as the C-M bound, 
\[k\leq \min_{t\in \mathbb{Z}_+} \{tr+k_{opt}^{(q)}(n-(r+1)t,d)\},\]
where $k_{opt}^{(q)}(n,d)$ is the largest possible dimension of an $[n,k,d]$ linear code over $\mathbb{F}_q$. Unfortunately, the value $k_{opt}^{(q)}(n,d)$ is not completely determined. 
For binary locally repairable codes, several explicit upper bounds on the dimension of linear locally repairable codes were given via a sphere-packing approach in \cite{AZD19}. 
Any locally repairable code achieving one of these bounds is called dimension-optimal, and many dimensional-optimal locally repairable codes have been constructed in \cite{HYUS16, LC21}. 
In particular, some propagation rules has been employed in \cite{CM15,AZD19}. The construction of locally repairable codes via concatenation was introduced by using a simple parity-check code as inner code \cite{CM15}, and a shortening technique to construct locally repairable codes was given in \cite[Lemma 10]{AZD19}. 

\subsection{Our Contributions and Techniques}
Among these constructions, propagation rules have not been extensively explored although various propagation rules have been discovered in classical coding theory. 
The current paper makes a step towards constructing locally repairable codes by exploring various propagation rules.
This paper makes the following six-fold contributions:
\begin{itemize}
\item[(i)] give two main propagation rules together with three other propagation rules for constructions of locally repairable codes. In addition, we show  that some trivial propagation rules in classical block codes do not hold anymore for locally repairable codes;
    \item[(ii)] construct many dimension-optimal binary locally repairable codes  based on the first main propagation rule, i.e., concatenating a locally repairable code as an inner code with a classical block code as an outer code. More dimension-optimal binary \LRCs are produced via minor propagation rules based on the aforementioned dimension-optimal locally repairable codes;
        \item[(iii)]  by concatenating locally repairable codes with algebraic geometry codes via the first main propagation rule, we are able to break the Zyablov-type bound for locally repairable codes;
            \item[(iv)]   produce a family of dimension-optimal binary locally repairable codes from extended Hamming codes based on the second main propagation rule, i.e.,  add certain rows and columns to a parity-check matrix of a given linear code. 
                    \item[(v)] convert Reed-Solomon codes into Singleton-optimal locally repairable codes via the second main propagation rule;
                             \item[(vi)] greatly  simplify the construction of a family of locally repairable codes given in \cite[Theorem 5]{MX20} that breaks the asymptotic Gilbert-Varshamov bound.
\end{itemize}

\subsection{Organization}
This paper is organized as follows. In Section \ref{sec:2}, we present some preliminaries including linear codes, algebraic geometry codes and locally repairable codes. In Section \ref{sec:3}, we introduce the first main propagation rule together with three other propagation rules, and then present many dimension-optimal locally repairable codes. In addition, we show that some trivial propagation rules for classical block codes do not hold anymore for locally repairable codes. Furthermore, we explicitly construct  a family of locally repairable codes that exceed the Zyablov-type bound. In Section \ref{sec:4}, we produce a family of dimension-optimal binary \LRCs from extended Hamming codes via the second main propagation rule, and convert a classical MDS code into a Singleton-optimal locally repairable code. Finally, we simplify a construction of a family of locally repairable codes which can break the asymptotic Gilbert-Varshamov bound.

\section{Preliminaries}\label{sec:2}
In this section, we present some preliminaries on the theory of linear codes, algebraic geometry codes and locally repairable codes.

\subsection{Linear codes}
In this subsection, we briefly discuss linear codes. The reader may refer to \cite{LX04,MS77} for more details.
Let $q$ be a prime power and $\mathbb{F}_q$ be the finite field with $q$ elements. Let $\mathbb{F}_q^n$ be the vector space of dimension $n$ over $\mathbb{F}_q$.
A linear code $\mC$ of length $n$ over $\mathbb{F}_q$ is an $\mathbb{F}_q$-subspace of $\mathbb{F}_q^n$. 
The dimension of $\mC$ is defined to be the dimension of $\mC$ as a vector space over $\F_q$. Any element in $\mC$ is called a codeword. 
The support of a codeword ${\bf u}=(u_1,\cdots, u_n)\in C$ is defined by $\text{supp}({\bf u})=\{i\in[n]: u_i\neq 0\},$ where $[n]=\{1,2,\cdots, n\}.$ 
The Hamming weight $\text{wt}(\bf u)$ of ${\bf u}$ is defined to be the size of $\text{supp}({\bf u}).$
 If $\mC\neq \{0\}$, then the minimum distance of $\mC$ is defined to be the smallest Hamming weight of nonzero codewords in $\mC$.

A $q$-ary linear code with length $n$, dimension $k$ and minimum distance $d$ is denoted as an $[n,k,d]_q$-linear code.
From the Singleton bound \cite[Theorem 5.4.1]{LX04}, we have the following inequality
\begin{equation}\label{eq:1}
d\leq n-k+1.
\end{equation}
A linear code with minimum distance achieving this Singleton bound \eqref{eq:1} is called an maximum distance separable code (MDS code for short).
The dual $\mC^\perp$ of any code $\mC$ is the orthogonal complement of $\mC$ in $\F_q^n$, i.e.,
$\mC^\perp:=\{\bx \in \F_q^n: \langle \bx,\bc\rangle =0 \text{ for any } \bc \in \mC\},$
where $\langle \cdot,\cdot\rangle$ is the canonical Euclidean inner product in $\F_q^n$.

\subsection{Algebraic geometry codes}\label{sec:2.2}
Let $F/\F_q$ be a function field with genus $g$ over the full constant field $\F_q$.  Let $\PP_F$ denote the set of places of $F$. Any place of $F$ with degree one is called rational.
For a divisor $G$ of function field $F/\F_q$, we define the Riemann-Roch space by
$\mL(G):=\{u\in F^*:\; (u)+G\ge 0\}\cup\{0\}.$
From Riemann's theorem, $\mathcal{L}(G)$ is a finite-dimensional vector space over $\F_q$ and its dimension $\ell(G)$ is lower bounded by $\ell(G)\ge \deg(G)-g+1$. Moreover, the equality holds true if $\deg(G)\ge 2g-1$ from \cite[Theorem 1.5.17]{St09}.

Let $\mP=\{P_1,\dots,P_n\}$ be a set of distinct rational places of $F$.
For a special divisor $G$ of $F$ with $0<\deg(G)<n$ and $\supp(G)\cap\mP=\emptyset$, the algebraic geometry code associated with $\mP$ and $G$ is defined to be
$\mC(\mP,G):=\{(f(P_1),f(P_2),\dots,f(P_n)): \; f\in\mL(G)\}.$
Then $\mC(\mP,G)$ is an $[n,k,d]_q$-linear code with dimension $k=\ell(G)$ and minimum distance $d\ge n-\deg(G)$ from \cite[Theorem 2.2.2]{St09}.

Let $N_q(g)$ be the maximum number of rational places of function fields over $\F_q$ with genus $g$. 
The real number $A(q)$ defined by $A(q):=\limsup_{g\rightarrow \infty} \frac{N_q(g)}{g}$ is called the Ihara's constant. If $q$ is a square, then $A(q)=\sqrt{q}-1$ \cite{GS95}. 
The famous Tsfaman-Vl\u{a}du\c{t}-Zink bound (TVZ bound for short) can be found from \cite[Theorem 8.4.7]{St09}.

\begin{prop}\label{prop:2.2}
Let $q=\ell^2$ be a square of a prime power. Then for all $\delta$ with $0\le \delta\le 1-(\ell -1)^{-1}$, there exists a family of algebraic geometry codes such that its information rate $R$ and relative minimum distance $\delta$ satisfy
$R\ge 1-\frac{1}{\sqrt{q}-1}-\delta.$
\end{prop}

\subsection{Locally repairable codes}\label{sec:2.3}
Roughly speaking, a block code is said with locality $r$ if  every coordinate of any given codeword can be recovered by accessing at most $r$ other coordinates of this codeword. 
A formal definition of a locally repairable code with locality $r$ can be given as follows.

\begin{defn}
A $q$-ary code of length $n$ is called a locally repairable code (LRC for short) with locality $r$ if for any $i\in[n],$ there exists a subset $R\subseteq [n]\setminus\{i\}$ of size $r$ such that for any ${\bf c}=(c_1,\cdots, c_n)\in\mC,$ $c_i$ can be recovered by $\{c_j\}_{j\in R},$ i.e., for any $i\in[n],$ there exists a subset $R\subseteq [n]\setminus\{i\}$ of size $r$ such that for any ${\bf u, v}\in\mC, {\bf u}_{R\cup \{i\}}={\bf v}_{R\cup \{i\}}$ if and only if ${\bf u}_R={\bf v}_R.$ The set $R\cup \{i\}$ is called a recovery set of $i.$ Note that we include $i$ in the recovery set for convenience. 
\end{defn}

In this paper, we always consider linear locally repairable codes.  Thus, a linear \LRC over $\F_q$ of length $n$, dimension $k$, minimum distance $d$ and locality $r$ is denoted to be an $[n,k,d;r]_q$-locally repairable code.
For such a $q$-ary $[n,k,d;r]$-locally repairable code, the minimum distance of $\mC$ is upper bounded by
\begin{equation}\label{Singletonbound}
d\le n-k-\left\lceil \frac{k}{r} \right\rceil+2.
\end{equation}
A code achieving this bound \eqref{Singletonbound} is usually called an optimal locally repairable code. However, we refer it as a Singleton-optimal locally repairable code in this paper.

Recovery sets of a linear locally repairable code can be characterized by its dual code from \cite[Lemma 5]{GXY19}. The precise result is given as below.
\begin{lemma}\label{lem:2.2}
A subset $R\subseteq [n]$ containing $i$ is a recovery set at $i$ for a $q$-ary linear code of length $n$ if and only if there exists a codeword ${\bf c}\in \mC^{\bot}$ such that $i\in {\rm supp}({\bf c})\subseteq R.$
\end{lemma}

For a linear code ${\mC}$  and $r\ge 1$,  we define the set
$\mathfrak{R}_{\mC}(r):=\{\supp(\bc):\; \bc\in {\mC}^\perp,\; |\supp(\bc)|\le r+1\}.$
The following result is a straightforward corollary of Lemma \ref{lem:2.2}.

\begin{cor}\label{cor:2.7} 
A linear code ${\mC}$ has locality $r$ if and only if $[n]=\cup_{I\in \mathfrak{R}_{\mC}(r)}I$.
\end{cor}

In this paper, we will also study the asymptotic behavior of locally repairable codes. Let the locality $r$ and alphabet size $q$ be fixed, and let the dimension and minimum distance be proportional to the length $n$.
Let $R_q(r,\delta)$ denote the asymptotic bound on the rate of $q$-ary locally repairable codes with locality $r$ and relative minimum distance $\delta$, i.e., $$R_q(r,\delta)=\limsup_{n\rightarrow \infty} \frac{\log_qM_q(n,\lfloor\delta n\rfloor,r)}{n} ,$$
where $M_q(n,d,r)$ is the maximum size of  \LRCs of length $n$, minimum distance $d$ and locality $r$.

For $0\le \delta \le 1-q^{-1}$, the asymptotic Gilbert-Varshamov bound of locally repairable codes is given in \cite{TBF16} by
$$R_q(r,\delta)\ge 1-\min_{0<s\le 1} \Big{\{}\frac{1}{r+1} \log_q\Big([1+(q-1)s]^{r+1}+(q-1)(1-s)^{r+1}\Big)-\delta \log_qs\Big{\}}.$$
Barg {\it et al.} \cite{BTV17} gave a construction of  asymptotically good $q$-ary locally repairable codes with locality $r$ whose rate $R$ and relative distance $\delta$ satisfy
\begin{equation}\label{construction_l}
R \ge\frac{r}{r+1}\Big{(}1-\delta-\frac{3}{\sqrt{q}+1}\Big{)},\quad r=\sqrt{q}-1,
\end{equation}
and
\begin{equation}\label{construction_l+1}
R \ge \frac{r}{r+1}\Big{(}1-\delta-\frac{\sqrt{q}+r}{q-1}\Big{)},\quad (r+1)|(\sqrt{q}+1).
\end{equation}
Furthermore, it was shown in \cite{BTV17} that for some values $r$ and $q$, the bound \eqref{construction_l+1} exceeds the asymptotic Gilbert-Varshamov bound  for locally repairable codes.
 Li {\it et al.} \cite{LMX19b} generalized the idea given in \cite{BTV17} by considering more subgroups of automorphism groups of function fields in the Garcia-Stichtenoth tower \cite{GS95}. This construction allows more flexibility of locality. In particular, if $r+1=up^v$ with $u|\gcd(p^v-1, \sqrt{q}-1)$, then there exists a family of explicit $q$-ary linear locally repairable codes with locality $r$ whose rate $R$ and relative distance $\delta$ satisfy
\begin{equation}\label{eq:lmx17}R\ge \frac{r}{1+r}\Big{(}1-\delta-\frac{\sqrt{q}+r-1}{q-\sqrt{q}}\Big{)}.\end{equation}
In order to overcome the restrictions on alphabet size $q$ and locality $r$, the authors \cite{MX20} provided an explicit construction via parity-check matrices whose columns are formed by coefficients of local expansions of function fields in the Garcia-Stichtenoth tower. In particular, for any fixed $q$ and $r$, it has been proved in \cite{MX20} that
 \begin{equation}
R\ge \frac{r}{r+1}-\frac{1}{A(q)}\times \frac{r}{r+1}-\delta.
\end{equation}

\section{Locally Repairable Codes via Concatenations}\label{sec:3}
In this section, we present the first main propagation rule together with three other propagation rules and construct many dimension-optimal locally repairable codes. In addition, we show  that some trivial propagation rules for classical block codes do not hold anymore for locally repairable codes.

\subsection{Concatenations and dimension-optimal locally repairable codes}
Let us start with the first main  propagation rule by concatenating a locally repairable  code as an inner code with a classical block code as an outer code. 

\begin{theorem}~\label{cc}
Let the inner code $\mathcal{C}_{in}$ be a $q$-ary  $[n_1,k_1,d_1;r]$-locally repairable code and let the outer code $\mathcal{C}_{out}$ be a $q^{k_1}$-ary $[n_2,k_2,d_2]$-linear code, then one can concatenate $\mathcal{C}_{in}$ with  $\mathcal{C}_{out}$ to obtain  an $[n_1n_2,k_1k_2,\ge d_1d_2;r]_q$-locally repairable code $\mathcal{C}_{conc}$.
\end{theorem}

\begin{proof} 
Fix an $\mathbb{F}_q$-vector space isomorphism $\varphi$ between $\mathbb{F}_{q^{k_1}}$ and $\mathcal{C}_{in}.$ Define $\mathcal{C}_{conc}$ by
\[\mathcal{C}_{conc}=\{(\varphi(c_1),\varphi(c_2),\dots,\varphi(c_{n_2})):\; (c_1,c_2,\dots,c_{n_2})\in\mathcal{C}_{out}\}. \]
Then $\mathcal{C}_{conc}$ is an $[n_1n_2,k_1k_2,\ge d_1d_2]_q$-linear code from \cite[Theorem 6.3.1]{LX04}. Every position of a codeword of $\mathcal{C}_{conc}$ can be determined by $r$ other positions due to the fact that the inner code $\mathcal{C}_{in}$  has locality $r$.  Hence, $\mathcal{C}_{conc}$ has locality $r$.
\end{proof}

Although the concatenation technique is simple, it is quite powerful. By concatenating locally repairable  codes with classical block codes, we can construct many good binary linear locally repairable codes that achieve the upper bound given in \cite[Theorem 6]{AZD19}.

\begin{lemma}\label{AZD19}
For any binary $[n,k,d;r]$-locally repairable code with locality $r$ such that $d\ge 5$ and $2\le r\le \frac{n}{2}-2,$ it follows that
\begin{equation}\label{eq:AZD19}
k\le \left\lfloor\frac{rn}{r+1}-\min\left\{\log_2\left(1+\frac{rn}{2}\right),\frac{rn}{(r+1)(r+2)}\right\}\right\rfloor.
\end{equation}
\end{lemma}

A binary linear locally repairable code is called dimension-optimal if its dimension achieves the bound \eqref{eq:AZD19} in Lemma~\ref{AZD19}. 
Based on our concatenation given in Theorem \ref{cc}, we can provide many dimension-optimal binary locally repairable codes. 

\begin{ex}
{\rm Let $\mC_{in}$ be a single parity-check $[5, 4, 2]_2$ code. 
Let $\mC_{out}$ be a $[17, 15, 3]_{2^4}$ MDS code obtained from rational algebraic geometry codes. By Theorem~\ref{cc}, we obtain an $[85, 60, 6; 4]_2$-locally repairable code. On the other hand, the bound  \eqref{eq:AZD19}  in Lemma~\ref{AZD19} shows that every $[85, k, 6; 4]_2$-locally repairable code must obey $k\le 60.582$, i.e., $k\le 60$. This implies that an $[85, 60, 6; 4]_2$-locally repairable code is dimension-optimal.
Similarly, more dimension-optimal binary locally repairable codes via concatenations can be listed in the following table. 

\begin{center}
Table I\\
{Dimension-optimal Locally Repairable Codes via Concatenations}\\\medskip
\begin{tabular}{|c|c|c|}\hline \hline
Inner code&Outer code&Dimension-optimal binary LRCs\\ \hline
$[5, 4, 2;4]_2$&$[17,15,3]_{16}$&$[85,60,6;4]_2$\\ \hline
$[5, 4, 2;4]_2$&$[16,14,3]_{16}$&$[80,56,6;4]_2$\\ \hline
$[5, 4, 2;4]_2$&$[15,13,3]_{16}$&$[75,52,6;4]_2$\\ \hline
$[5, 4, 2;4]_2$&$[14,12,3]_{16}$&$[70,48,6;4]_2$\\ \hline
$[5, 4, 2;4]_2$&$[13,11,3]_{16}$&$[65,44,6;4]_2$\\ \hline
$[4,3,2;3]_2$&$[9,7,3]_{8}$&$[36,21,6;3]_2$\\ \hline
\end{tabular}
\end{center}}
\end{ex}

Some classical propagation rules given in \cite[Theorem 6.1.1]{LX04} can be generalized to construct new locally repairable codes from old ones. 
\begin{lemma}\label{lem:3.4} 
If $\mC$ is an $[n,k,d;r]_q$-locally repairable code, then
\begin{itemize}
\item[{\rm (i)}] there exists an $[n+1,k,d;r]_q$-locally repairable code;
\item[{\rm (ii)}] there exists an $[n-1,\ge k-1,d;r]_q$-locally repairable code;
\item[{\rm (iii)}] there exists an $[n-t,\ge k-t+s,\ge d-s;r]_q$-locally repairable code for any $0\le s\le t$, provided that $\mC$ has disjoint recovery sets and one of the recovery sets has size $t$.
\end{itemize}
\end{lemma}
 \begin{proof} (i) By adding  $0$ to the $(n+1)$-th position of every codeword of $\mC$, one gets an $[n+1,k,d]_q$-linear code. For the $(n+1)$-th  position, we have a recovery set $\{n+1\}$. Thus, the new code has locality $r$ as well.

 (ii) Let $H$ be a parity-check matrix of $\mC$. From \cite[Theorem 4.5.6]{LX04}, there are $d$ columns of $H$ which are linearly dependent. Without loss of generality, assume that the last column of $H$ is not in these $d$ columns. We delete the last column of $H$ to form a $(n-k)\times(n-1)$ matrix $H_1$. Let $\mC_1$ be the code with $H_1$ as its parity-check matrix. It is clear that $\mC_1$ is an $[n-1,\ge k-1, d]_q$-linear code. To obtain the locality, for each $i\in[n-1]$, there exists a codeword $\bu\in\mC^\perp$ such that $i\in\supp(\bu)$ and $|\supp(\bu)|\le r+1$. Let $\bu_1$ be the vector obtained from $\bu$ by deleting the last position. Then $\bu_1\in\mC_1^\perp$, $i\in\supp(\bu_1)$ and $|\supp(\bu_1)|\le r+1$. This implies that the position $i$ has a recovery set of size at most $r+1$. Hence, $\mC_1$ is an $[n-1,\ge k-1,d;r]_q$-locally repairable code.

 (iii) Without loss of generality, we may assume that $\{n-t+1,n-t+2,\dots,n\}$ is a recovery set of size $t$.  Let $H$ be a parity-check matrix of $\mC$. We delete the last $t-s$ columns of $H$ to form an $(n-k)\times(n-t+s)$ matrix $H_2$. Let $\mC_2$ be the code with $H_2$ as its parity-check matrix. It is clear that $\mC_2$ is an $[n-t+s,\ge k-t+s,\ge d]_q$-linear code. As in (ii), we can show that $\mC_2$ has locality $r$ as well. Furthermore, $\{n-t+1,n-t+2,\cdots, n-t+s\}$ is a recovery set that is disjoint with other recovery sets. Now we delete the last $s$ positions of $\mC_2$ to obtain $\mC_3$. It is easy to see that $\mC_3$ is an $[n-t,\ge k-t+s,\ge d-s]_q$-linear code. As a whole recovery set of $\mC_2$ is deleted, $\mC_3$ has  disjoint recovery sets with each size being at most $r+1$, i.e., $\mC_3$ has locality $r$.
 \end{proof}

\begin{rem}{\rm For classical block codes, we have the following propagation rules: (i) an $[n,k,d]_q$-linear code gives an $[n-1,k,d-1]_q$-linear code;
(ii) an $[n,k,d]_q$-linear code gives an $[n,k,d-1]_q$-linear code. However, these two propagation rules does not hold anymore. Namely, (i) an $[n,k,d;r]_q$-locally repairable code does not always produce an $[n-1,k,d-1; r]_q$-locally repairable code;
(ii) an $[n,k,d;r]_q$-locally repairable code does not always produce an $[n,k,d-1; r]_q$-locally repairable code. To see this, let us give two counter-examples.

Counter-example 1: By Table I, we have a binary $[85,60,6;4]$-locally repairable code. Suppose we had a binary $[85-1,60,6-1;4]$-locally repairable code. Then by Lemma \ref{AZD19}, any $[84,k,5;4]_2$-\LRC must satisfy $k\le 59.8$. This is a contradiction.

Counter-example 2: Consider the binary $[5,4,2;4]$-locally repairable code. As we do not have a \LRC with minimum distance $1$, we have no way to get a $[5,4,1;4]$-locally repairable code.
}
\end{rem}

\begin{ex}{\rm 
We can make use of the propagation rules given in Lemma \ref{lem:3.4} to construct more dimension-optimal \LRCs from Table I. 
\begin{center}
Table II\\{Dimension-optimal Locally Repairable Codes via Propagation Rules in Lemma \ref{lem:3.4}} \medskip
\vspace{0.2cm}
\begin{tabular}{|c|c|c|}\hline \hline
Codes from Table I&Dimension-optimal LRCs&Propagation rules\\ \hline
$[85,60,6;4]_2$&$[84,59,6;4]_2$&Lemma \ref{lem:3.4}(ii)\\ \hline
$[80,56,6;4]_2$&$[79,55,6;4]_2$&Lemma \ref{lem:3.4}(ii)\\ \hline
$[75,52,6;4]_2$&$[74,51,6;4]_2$&Lemma \ref{lem:3.4}(ii)\\ \hline
$[85,60,6;4]_2$&$[80,56,\ge 5;4]_2$&Lemma \ref{lem:3.4}(iii)\\ \hline
$[80,56,6;4]_2$&$[75,52,\ge 5;4]_2$&Lemma \ref{lem:3.4}(iii)\\ \hline
$[75,52,6;4]_2$&$[70,48,\ge 5;4]_2$&Lemma \ref{lem:3.4}(iii)\\ \hline
$[70,48,6;4]_2$&$[65,44,\ge 5;4]_2$&Lemma \ref{lem:3.4}(iii)\\ \hline\hline
\end{tabular}
\end{center}
\vspace{0.2cm}
}\end{ex}

\subsection{Zyablov-type bound}
In classical coding theory, in order to obtain an explicit asymptotic bound, one can concatenate a family of linear codes achieving the Gilbert-Varshamov bound as an inner code and a Reed-Solomon code as an outer code. The explicit asymptotic bound obtained in this way is called the Zyablov bound. In this subsection, we will explore the same technique to obtain the Zyablov-type bound for locally repairable codes.

Consider the inner code $\mC_{\rm in}$ to be an $[n_1,k_1,d_1; r]_2$-locally repairable code that achieves the Gilbert-Varshamov bound, and the outer code $\mC_{\rm out}$ to be an $[n_2,k_2,d_2=n_2-k_2+1]_q$ Reed-Solomon code with $q=2^{k_1}$. The concatenated code $\mC_{\rm conc}$ is a binary $[n_1 n_2,k_1 k_2,\ge d_1 d_2;r]$-locally repairable code from Theorem \ref{cc}. In particular, its rate satisfies

\begin{equation}\label{eq:3.1}
R=\frac{k_1}{n_1}\times\frac{k_2}{n_2}=\left(1-h\left(r,\frac{d_1}{n_1}\right)+o(1)\right)\left(1-\frac{d_2}{n_2}\right),
\end{equation}
where
\begin{equation}\label{eq:3.2}
h(r,x):=\min_{0<s\le 1} \Big{\{}\frac{1}{r+1} \log_q\Big([1+(q-1)s]^{r+1}+(q-1)(1-s)^{r+1}\Big)-x \log_qs\Big{\}}.
\end{equation}
Put $\tau=\frac{d_1}{n_1}$ and $\Gd=\frac{d_1d_2}{n_1n_2}$. Then we have $\frac{d_2}{n_2}=\frac{\Gd}{\tau}$. Substituting $\frac{d_1}{n_1}=\tau$ and
$\frac{d_2}{n_2}=\frac{\Gd}{\tau}$ into Equation \eqref{eq:3.1}, we obtain the following Zyablov-type bound for locally repairable codes.

\begin{theorem}(Zyablov-type bound)
For given integer $r\ge 1$ and real $\Gd\in (0,1)$,
there exists a family of binary locally repairable codes of rate $R$, relative minimum distance $\delta$ and locality $r$ satisfying
$\displaystyle R = \max_{\Gd\leq \tau\leq 1} (1-\delta/\tau)(1-h(r,\tau)).$
Furthermore, this family of locally repairable codes can be constructed in polynomial time.
\end{theorem}

\begin{rmk}
The Zyablov-type bound is worse than the Gilbert-Varshamov bound, since $h(r,\tau)$ is a decreasing function in the variable $\tau$ and  
\[ \max_{\Gd\leq \tau\leq 1} (1-\delta/\tau)(1-h(r,\tau))\le \max_{\Gd\leq \tau\leq 1} (1-h(r,\tau))\le 1-h(r,\Gd).\]
 \end{rmk}

\subsection{Breaking the Zyablov-type bound}
In this subsection, we provide an explicit  construction of asymptotically good locally repairable codes exceeding the Zyablov-type bound via concatenated codes.

\begin{prop}\label{prop:3.8}
{\rm Let $\mC_{in}$ be the extended Hamming code $[8,4,4]_2$ with locality $3$.
Let $\{\mC_{out}^{(i)}\}_{i=1}^\infty$ be a family of algebraic geometry codes $[n_i,k_i,d_i]_{2^4}$ attaining the TVZ bound given in Proposition \ref{prop:2.2}.
Then the concatenated codes  are a family of $[8n_i,4k_i,4d_i; 3]_{2}$-locally repairable codes with locality $3$ such that its information rate $R$ and relative minimum distance $\delta$ satisfy
$R+\delta\ge 1/3.$}
\end{prop} 
\begin{proof}
From Theorem \ref{cc}, the concatenated codes are a family of $[8n_i,4k_i,4d_i;3]_{2}$-locally repairable codes with locality $3$. Thus, we have
$$R+\delta=\lim_{i\rightarrow \infty} \frac{4k_i+4d_i}{8n_i} \ge \frac{1}{2}-\frac{1}{2\cdot A(2^4)}=\frac{1}{2}-\frac{1}{2\times 3}=\frac{1}{3}.$$
\end{proof}

\begin{rmk} Proposition \ref{prop:3.8} gives
 an explicit family of \LRCs with locality $3$ that are beyond the Zyablov-type bound in some interval (see Figure \ref{844}). 
\begin{figure}[H]~\label{844.png}
\begin{center}
\includegraphics[width=19em]{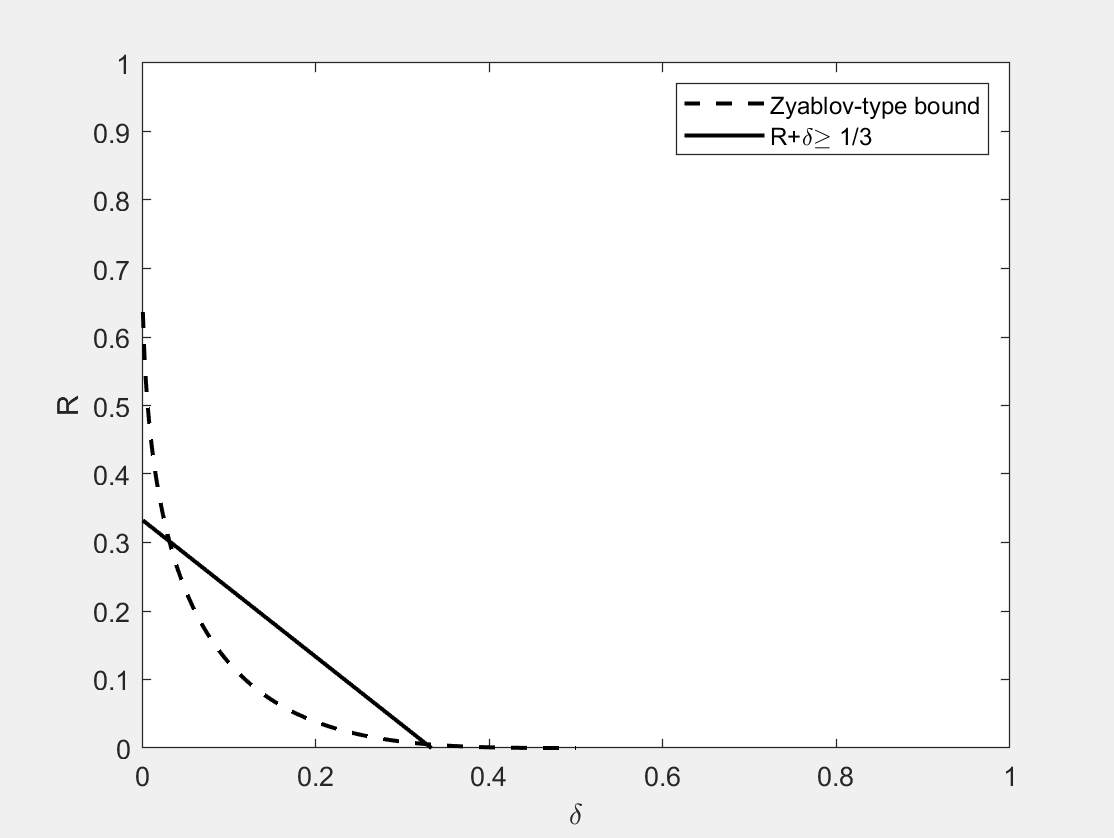}
\end{center}
\caption{Beyond the Zyablov-type bound for locality $3$}~\label{844}
\end{figure}
\end{rmk}

\begin{prop}\label{prop:3.9}
{\rm Let $\mC_{in}$ be the extended Golay code $[24,12,8]_2$ with locality $7$.
Let $\{\mC_{out}^{(i)}\}_{i=1}^\infty$ be a family of algebraic geometry codes $[n_i,k_i,d_i]$ over $\F_{2^{12}}$ attaining the TVZ bound given in Proposition \ref{prop:2.2}.
Then the concatenated codes  are a family of $[24n_i,12k_i,8d_i]_{2}$-linear locally repairable codes with locality $7$ such that its information rate $R$ and relative  minimum distance $\delta$ satisfy
$R\ge {31}/{63}-3\delta/2.$}
\end{prop}
\begin{rmk}
Proposition \ref{prop:3.9} gives
 an explicit family of \LRCs with locality $7$ that are beyond the Zyablov-type bound in some interval (see  Figure~\ref{24128}).
\begin{figure}[H]\label{24128.png}
\begin{center}
\includegraphics[width=19em]{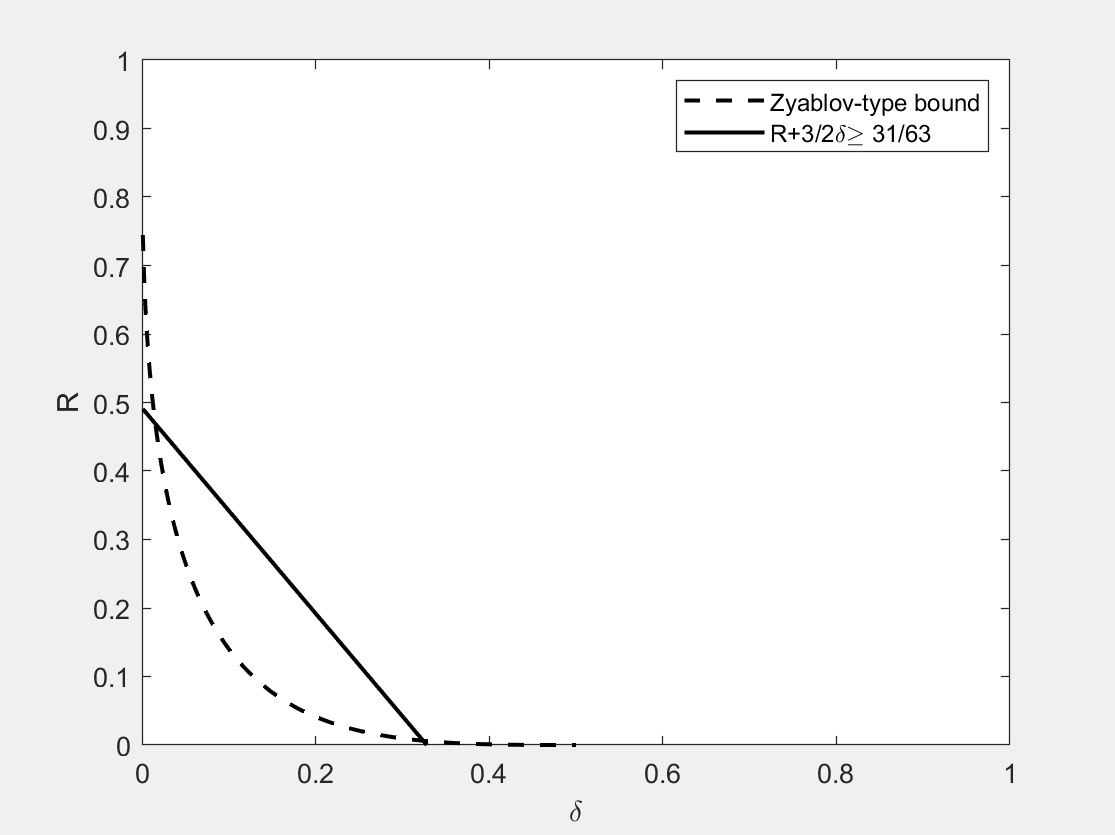}
\end{center}
\vspace{-0.25 cm}
\caption{Beyond the Zyablov-type bound for locality $7$}~\label{24128}
\end{figure}
\end{rmk}

\begin{prop}\label{prop:3.19}
{\rm Let $m\le q$ and $1\le t\le m-1$ be positive integers.
Let $\mC_{in}$ be an $[m,m-t,t+1]_q$ Reed-Solomon code.
Let $\{\mC_{out}^{(i)}\}$ be a family of algebraic geometry codes $[n_i,k_i,d_i]$ over $\F_{q^{m-t}}$ attaining the TVZ bound given in Proposition \ref{prop:2.2}.
Then the concatenated codes are \LRCs  with locality $m-t$ such that its information rate $R$ and relative  minimum distance $\delta$ satisfy
$$R \ge  \left(1-\frac{t}{m}\right) \left(1-\frac{1}{A(q^{m-t})}\right)-\frac{m-t}{t+1}\delta.$$
Furthermore, such codes can be explicitly constructed.}
\end{prop}
\begin{proof} Let us consider the $[m,m-t,t+1]_q$ Reed-Solomon code given by
\[\mA:=\{(f(\Ga_1),\dots,f(\Ga_m)):\; f(x)\in\F_q[x]_{<m-t}\},\]
where $\Ga_1,\dots,\Ga_m$ are pairwise distinct elements of $\F_q$. 
From Lagrange interpolation formula, it is clear that the $[m,m-t,t+1]_q$ Reed-Solomon code has locality $m-t$.
From  Theorem \ref{cc},  the concatenated codes have parameters $[mn_i,(m-t)k_i,(t+1)d_i;m-t]_q$. Hence, we have
$(t+1)(m-t)k_i+(m-t)(t+1)d_i\ge  (m-t)(t+1)(k_i+d_i).$
This gives
 $$(t+1)R+(m-t)\delta \ge  \left(1-\frac{t}{m}\right)(t+1)\left(1-\frac{1}{A(q^{m-t})}\right).$$
 The desired result follows.

 \end{proof}

By taking $m=r+1$ and $t=1$ in Proposition \ref{prop:3.19}, we obtain the following result.
\begin{cor}~\label{cor:3.11}
For any given real number $\Gd\in(0,1)$, there exists an explicit family of \LRCs with rate $R$, locality $r$ and relative minimum distance $\Gd$ satisfying
$$ R \ge  \frac{r}{r+1} \left(1-\frac{1}{A(q^{r})}\right)-\frac{r}{2}\delta.$$
\end{cor}

In particular, if $q$ is a prime, then we can obtain the same asymptotic bound of locally repairable codes over a prime finite field  as given in \cite[Theorem 10]{MX20}.

\section{Locally repairable codes via lengthening parity-check matrices}\label{sec:4}
In this section, we try to construct good locally repairable codes based on linear codes from a lengthening propagation rule.
In particular, the technique is to endow the locality by adding rows and columns in the parity-check matrix of linear codes.

Let $\mC_0$ be an $[n,k,d]$-linear code over $\F_q$ and let $H_0\in \F_q^{(n-k)\times n}$ be its parity-check matrix.
Let $r$ be a positive integer. There exist integers $m$ and $s$ such that $n=mr+s$ with $1\le s\le r$.
We divide the columns of $H_0$ into $\lceil \frac{n}{r}\rceil$ blocks.
Let ${\bf h}_{ij}$ be the columns of $H_0$ for $1\le i\le m, 1\le j\le r$ and $i=m+1,1\le j\le s$, i.e.,
$$H_0=(\bh_{11}, \cdots,  \bh_{1r},\bh_{21}, \cdots,  \bh_{2r},\cdots, \bh_{m1}, \cdots,  \bh_{mr},  \bh_{m+1,1}, \cdots, \bh_{m+1,s}).$$
Now let us consider the matrix $H$ defined by
\begin{equation}\label{eq:4.11} H:=\left(\begin{array}{cccc|c|cccc|cccc}
1 & \cdots & 1   & 1 & \cdots  & 0 & \cdots & 0 &  0 & 0 & \cdots & 0& 0\\
\vdots&\ddots&\vdots&\vdots & \ddots & \vdots & \ddots & \vdots & \vdots & \vdots & \ddots &\vdots & \vdots\\
 0 & \cdots & 0   &  0 & \cdots  & 1 & \cdots & 1 &  1 & 0 & \cdots & 0 & 0\\
 0 & \cdots & 0   &  0 & \cdots & 0 & \cdots & 0 & 0 & 1 & \cdots & 1 & 1\\
 \bh_{11} & \cdots & \bh_{1r}  & \bo &  \cdots &\bh_{m1} & \cdots & \bh_{mr} & \bo & \bh_{m+1,1} & \cdots & \bh_{m+1,s} & \bo
\end{array}
\right),
\end{equation}
where $\bo$ stands for the zero vector of dimension $n-k$.

\begin{lemma}\label{lem:4.1}
Let $H$ be the matrix defined in Equation \eqref{eq:4.11}. Then any $d-1$ columns of $H$ are linearly independent.
\end{lemma}
\begin{proof}
Let $\bc_{i,j}$ denote the $j$-th column in the $i$-th block of the parity-check matrix $H$.
Choose arbitrary $d-1$ columns $\{\bc_{i,j}: i\in I,j\in S_i\}$ from $H$, where $S_i$ is a nonempty subset of $\{1,2,\cdots,r+1\}$ with $\sum_{i\in I} |S_i|=d-1$.

{\bf Case 1:} the last column in each block does not belong to $\{\bc_{i,j}: i\in I,j\in S_i\}$, i.e., $r+1\notin S_i$ for $1\le i\le m$ and $s+1\notin S_{m+1}$.
Since $H_0$ is a parity-check matrix of an $[n,k,d]$-linear code, any $d-1$ columns of $H_0$ are linearly independent.
Hence, any $d-1$ columns $\{\bc_{i,j}: i\in I,j\in S_i\}$ are linearly independent from linear algebra.

{\bf Case 2:}  the last column of some blocks are chosen. 
Without loss of generality, assume that $r+1\in S_{i_0}$ for some $1\le i_0\le m$.
In this case, we can prove a stronger result. That is, arbitrary $d$ columns $\{\bc_{i,j}: i\in I,j\in S_i\}$ are linearly independent. 
Assume that there exist $\Gl_{i,j}\in \F_q$  such that $$\sum_{i_0\neq i\in I}\sum_{j\in S_i} \Gl_{ij} \bc_{ij}+\sum_{j\in S_{i_0}} \Gl_{i_0,j} \bc_{i_0,j}=\bo.$$
From the first $\lceil\frac{n}{r}\rceil$ rows of $H$, we have $\sum_{j\in S_{i}} \Gl_{i,j}=0$ for $i\in I$.
It follows that
\begin{equation}\label{eq:lem4.1} \Gl_{i_0,r+1}=-\sum_{j\in S_{i_0}\setminus \{r+1\}} \Gl_{i_0,j}.\end{equation}
Then we have
$$\sum_{i_0\neq i\in I} \sum_{j\in S_i}\Gl_{ij} \bc_{ij}+\sum_{j\in S_{i_0}\setminus \{r+1\}} \Gl_{i_0,j} (\bc_{i_0,j}-\bc_{i_0,r+1})=\bo.$$
By focusing on the last $n-k$ rows of $H$, we have $$\sum_{i_0\neq i\in I} \sum_{j\in S_i} \Gl_{ij} \bh_{ij}+\sum_{j\in S_{i_0}\setminus \{r+1\}} \Gl_{i_0,j} \bh_{i_0,j}=\bo.$$
Since any $d-1$ columns of $H_0$ are linearly independent, we have $\Gl_{i,j}=0$ for $i_0\neq i\in I,j\in S_i$ and $i=i_0,j\in S_{i_0}\setminus \{r+1\}$.
From Equation \eqref{eq:lem4.1}, we have $\Gl_{i_0,r+1}=0$. Hence, any $d$ columns are linearly independent in this case.
\end{proof}

\begin{prop}\label{prop:4.2}
Let $\mC$ be the code with $H$ given in Equation \eqref{eq:4.11} as its parity-check matrix. Then $\mC$ is an $[n+\lceil \frac{n}{r}\rceil,k,\ge d; r]$-locally repairable code.
\end{prop}
\begin{proof} 
Since $H$ is an $(n-k+\lceil \frac{n}{r}\rceil)\times (n+\lceil \frac{n}{r}\rceil)$ matrix and the rows of $H$ are linearly independent from linear algebra, the dimension of $\mC$ is $(n+\lceil \frac{n}{r}\rceil)-(n-k+\lceil \frac{n}{r}\rceil)=k$.
From Lemma \ref{lem:4.1}, any $d-1$ columns of $H$ are linearly independent. 
Hence, the minimum distance of $\mC$ is at least $d$ from \cite[Theorem 4.5.6]{LX04}.
Moreover, the locality of $\mC$ is $r$ from the parity-check matrix $H$ and Corollary \ref{cor:2.7}.  

\end{proof}

\subsection{Dimension-optimal locally repairable codes}
In order to increase minimum distance of locally repairable codes constructed in Proposition \ref{prop:4.2}, we can employ the additional property of parity-check matrices of extended Hamming codes. In this subsection, we propose a construction of dimension-optimal locally repairable codes via lengthening extended Hamming codes.

\begin{theorem}\label{thm:4.3}
Let $\mC_0$ be a binary $[2^t, 2^t-1-t,4]$-extended Hamming code for $t\ge 3$. Let $H_0$ be a parity-check matrix of $\mC_0$ and let $H$ be the matrix given in Equation \eqref{eq:4.11}. Then, for $r=2$ or $3$, the code $\mC$ with $H$ as a parity-check matrix is a binary  $[2^t+\lceil\frac{2^t}r\rceil, 2^t-1-t,5;r]$-locally repairable code. Furthermore, $\mC$ achieves the upper bound \eqref{eq:AZD19} if $(r=2,t\ge 5)$ or $(r=3,t\ge 5)$, i.e., $\mC$ is dimension-optimal in these cases.
\end{theorem}
\begin{proof} Let us first show that $\mC$ is a $[2^t+\lceil\frac{2^t}r\rceil, 2^t-1-t,5;r]$-locally repairable code. By Proposition \ref{prop:4.2}, it is sufficient to show that any four columns of $H$ are linearly independent. Let us prove this only for the case where $r=3$.

Case 1: all four columns belong to the same block, say the first block. In this case, we have to show that $(1,\bo,\bo)^T$, $(1,\bo,\bh_{11})^T$, $(1,\bo,\bh_{12})^T$ and $(1,\bo,\bh_{13})^T$ are linearly independent. Suppose that
\begin{equation}\label{eq:4.1}
\Gl_0(1,\bo,\bo)^T+\sum_{i=1}^3\Gl_i(1,\bo,\bh_{1i})^T=\bo^T
 \end{equation}
 for some $\Gl_i\in\F_2$ with $0\le i\le 3$. Then we have $\sum_{i=1}^3\Gl_i\bh_{1i}=\bo$. This forces that $\Gl_1=\Gl_2=\Gl_3=0$ from \cite[Theorem 4.5.6]{LX04}, since $\mC_0$ has minimum distance $4$. Thus, it follows from \eqref{eq:4.1} that $\Gl_0=0$.

Case 2: Other possibilities: (a) three columns belong to one block and one column belongs to another block; (b) two columns belong to one block and other two belong to another block; (c) two columns belong to one block and other two columns belong to other two blocks, respectively; (d) four columns belong to four distinct blocks. For these four possibilities, one can use a similar argument to show that they are linearly independent. This proves the first part.

 Now if $r=2$ and $t\ge 5$, we want to show that
 \begin{equation}\label{eq:4.2}
k=2^t-1-t\ge \left\lfloor\frac{rn}{r+1}-\min\left\{\log_2\left(1+\frac{rn}{2}\right),\frac{rn}{(r+1)(r+2)}\right\}\right\rfloor
 \end{equation}
 with $r=2$ and $n=2^t+2^{t-1}$.
To show \eqref{eq:4.2}, it will be sufficient to show the following two inequalities
 \begin{equation}\label{eq:4.3}
2^t-t>\frac{rn}{r+1}-\log_2\left(1+\frac{rn}{2}\right)
 \end{equation}
 and
  \begin{equation}\label{eq:4.4}
2^t-t>\frac{rn}{r+1}-\frac{rn}{(r+1)(r+2)}.
 \end{equation}
 
 Substituting $r=2$ and $n=2^t+2^{t-1}$ into \eqref{eq:4.3} and rewrite it into the following form
 \begin{equation}\label{eq:4.5}
\log_2\left(1+2^t+2^{t-1}\right)>t.
 \end{equation}
 It is clear that \eqref{eq:4.5} holds and hence \eqref{eq:4.3} holds.

 Substituting $r=2$ and $n=2^t+2^{t-1}$ into \eqref{eq:4.4} and rewrite it into the following form
 \begin{equation}\label{eq:4.6}
2^t-t>\frac{r}{r+2}\times n=2^{t-1}+2^{t-2}.
 \end{equation}
 The inequality \eqref{eq:4.6} holds when $t\ge 5$ and hence \eqref{eq:4.4} holds. This completes the proof for the case where $r=2$ and $t\ge 5$.

Next assume that $r=3$ and $t\ge 5$. In this case we want to show that the two inequalities \eqref{eq:4.3}  and \eqref{eq:4.4}
 hold for $r=3$ and $n=2^t+\lceil\frac{2^{t}}{3}\rceil$.

 Substituting $r=3$ and $n=2^t+\left\lceil\frac{2^{t}}{3}\right\rceil$ into \eqref{eq:4.3} and rewrite it into the following form
 \begin{equation}\label{eq:4.7}
\log_2\left(1+\frac32\left(2^t+\left\lceil\frac{2^{t}}{3}\right\rceil\right)\right)>t+\frac34\times \left\lceil\frac{2^{t}}{3}\right\rceil-\frac{2^{t}}{4}.
 \end{equation}
 Note that we have $1+\frac32\left(2^t+\left\lceil\frac{2^{t}}{3}\right\rceil\right)>\frac32\times2^t>\sqrt{2}\times 2^t$. This gives
 \[\log_2\left(1+\frac32\left(2^t+\left\lceil\frac{2^{t}}{3}\right\rceil\right)\right)>t+\frac12=t+\frac34\times \frac{2^{t}+2}{3}-\frac{2^{t}}{4}\ge
 t+\frac34\times \left\lceil\frac{2^{t}}{3}\right\rceil-\frac{2^{t}}{4},\]
 i.e.,  the inequality \eqref{eq:4.7} holds  and hence \eqref{eq:4.3} holds.

 Substituting $r=3$ and $n=2^t+\left\lceil\frac{2^{t}}{3}\right\rceil$  into \eqref{eq:4.4} and rewrite it into the following form
 \begin{equation}\label{eq:4.8}
2^t-t>\frac{r}{r+2}\times n=\frac35\left(2^t+\left\lceil\frac{2^{t}}{3}\right\rceil\right)
 \end{equation}
 The inequality \eqref{eq:4.8} holds when $t\ge 5$ and hence \eqref{eq:4.4} holds. This completes the proof for the case where $r=3$ and $t\ge 5$.
\end{proof}

Theorem \ref{thm:4.3} produces two families of dimension-optimal \LRCs based on our second propagation rule and extended Hamming codes. Combining with the propagation rules given in Lemma \ref{lem:3.4}, we can obtain more dimension-optimal locally repairable codes which are listed in the following table. 
\begin{center}
Table III\\
{ Dimension-optimal LRCs via Lengthening and Propagation Rules in Lemma \ref{lem:3.4}}\\\medskip
\begin{tabular}{|c|c|c|c|}\hline \hline
$(r,t)$ & Known optimal LRCs &New optimal LRCs & Propagation rules\\
& {given in Theorem \ref{thm:4.3}}& &\\ \hline
$(2,5)$& $[48,26,5;2]_2$&$[47,25,5;2]_2$&Lemma \ref{lem:3.4}(ii)\\ \hline
$(2,6)$& $[96,57,5;2]_2$&$[95,56,5;2]_2$&Lemma \ref{lem:3.4}(ii)\\ \hline
$(2,7)$& $[192,120,5;2]_2$&$[191,119,5;2]_2$&Lemma \ref{lem:3.4}(ii)\\ \hline
$(2,8)$& $[384,247,5;2]_2$&$[383,246,5;2]_2$&Lemma \ref{lem:3.4}(ii)\\ \hline
$(2,9)$& $[768,502,5;2]_2$&$[767,501,5;2]_2$&Lemma \ref{lem:3.4}(ii)\\ \hline\hline

$(3,5)$& $[43,26,5;3]_2$&$[42,25,5;3]_2$&Lemma \ref{lem:3.4}(ii)\\ \hline
$(3,6)$& $[86,57,5;3]_2$&$[85,56,5;3]_2$&Lemma \ref{lem:3.4}(ii)\\ \hline
$(3,7)$& $[171,120,5;3]_2$&$[170,119,5;3]_2$&Lemma \ref{lem:3.4}(ii)\\ \hline
$(3,8)$& $[342,247,5;3]_2$&$[341,246,5;3]_2$&Lemma \ref{lem:3.4}(ii)\\ \hline
$(3,9)$& $[683,502,5;3]_2$&$[682,501,5;3]_2$&Lemma \ref{lem:3.4}(ii)\\ \hline\hline

$(3,5)$& $[43,26,5;3]_2$&$[44,26,5;3]_2$&Lemma \ref{lem:3.4}(i)\\ \hline
$(3,7)$& $[171,120,5;3]_2$&$[172,120,5;3]_2$&Lemma \ref{lem:3.4}(i)\\ \hline
$(3,9)$& $[683,502,5;3]_2$&$[684,502,5;3]_2$&Lemma \ref{lem:3.4}(i)\\ \hline\hline
\end{tabular}
\end{center}

\subsection{Singleton-optimal locally repairable codes via lengthening RS codes}
Again, in order to increase the minimum distance of locally repairable codes given in Proposition \ref{prop:4.2}, we can employ the additional property of parity-check matrices with  the Vandermonde structure as given in \cite{J19,XY19}.
In this subsection, let us consider lengthening Reed-Solomon codes. 

Let $r$ be a positive integer and let $n\le q-1$ be a positive integer. There exist integers $m$ and $s$ such that $n=mr+s$ with $1\le s\le r$.
Let $d\ge 2$ be a positive integer.
Let $H_0$ be a $(d-1)\times n$ matrix consisting of columns ${\bf h}_{ij}=(\Ga_{ij},\Ga_{ij}^2,\cdots,\Ga_{ij}^{d-1})^T$ with pairwise distinct $\Ga_{ij}\in \F_q^*$ for $1\le i\le m$, $1\le j\le r$ and $i=m+1,1\le j\le s$.
Any $d-1$ columns of $H_0$ are linearly independent, since the determinant of Vandermonde matrix is nonzero.
Let $H_1$ be a $d\times n$ matrix given as follows:
$$H_1=\left(\begin{array}{cccccccccc}
1 & \cdots & 1   & \cdots  & 1 & \cdots & 1 &   1 & \cdots & 1\\
 \bh_{11} & \cdots & \bh_{1r}  &  \cdots &\bh_{m1} & \cdots & \bh_{mr} &  \bh_{m+1,1} & \cdots & \bh_{m+1,s}
\end{array}
\right).$$
It is easy to see that $H_1$ is a Vandermonde matrix which is a generator matrix of some Reed-Solomon code.
Let $\mC_1$ be the code with $H_1$ as its parity-check matrix. Then $\mC_1$ is an $[n, n-d,d+1]$ MDS code.
We split the first row of $H_1$ into $\lceil \frac{n}{r}\rceil$ rows and add a new column in each block of $H_1$ as follows:
\begin{equation}\label{eq:lem4.5} H=\left(\begin{array}{cccc|c|cccc|cccc}
1 & \cdots & 1   & 1 & \cdots  & 0 & \cdots & 0 &  0 & 0 & \cdots & 0& 0\\
\vdots&\ddots&\vdots&\vdots & \ddots & \vdots & \ddots & \vdots & \vdots & \vdots & \ddots &\vdots & \vdots\\
 0 & \cdots & 0   &  0 & \cdots  & 1 & \cdots & 1 &  1 & 0 & \cdots & 0 & 0\\
 0 & \cdots & 0   &  0 & \cdots & 0 & \cdots & 0 & 0 & 1 & \cdots & 1 & 1\\
 \bh_{11} & \cdots & \bh_{1r}  & \bo &  \cdots &\bh_{m1} & \cdots & \bh_{mr} & \bo & \bh_{m+1,1} & \cdots & \bh_{m+1,s} & \bo
\end{array}
\right),
\end{equation}
where $\bo$ stands for the zero vector of dimension $d-1$.

\begin{lemma}\label{lem:4.5}
Any $d$ columns of $H$ given in Equation \eqref{eq:lem4.5} are linearly independent.
\end{lemma}
\begin{proof}
Let $\bc_{i,j}$ be the $j$-th column in the $i$-th block of the parity-check matrix $H$.
Choose any $d$ columns $\{\bc_{i,j}: i\in I,j\in S_i\}$ from $H$, where $S_i$ is a nonempty subset of $\{1,2,\cdots,r+1\}$ with $\sum_{i\in I} |S_i|=d$.

{\bf Case 1:} the last column in each block does not belong to $\{\bc_{i,j}: i\in I,j\in S_i\}$, i.e., $r+1\notin S_i$ for $1\le i\le m$ and $s+1\notin S_{m+1}$.
Assume that there exist  $\Gl_{i,j}\in \F_q$ such that $\sum_{i\in I}\sum_{j\in S_i} \Gl_{ij} \bc_{i,j}=\bo.$
It follows that $\sum_{j\in S_i} \Gl_{i,j}=0$  for all $i\in I$ from the parity-check matrix $H$, which imply that $\sum_{i\in I}\sum_{j\in S_i} \Gl_{i,j}=0$.
Thus, we have the following system of linear equations
$$\sum_{i\in I, j\in S_i}  \Ga_{ij}^t \Gl_{i,j} =0  \text{ for } 0\le t \le d-1.$$
The coefficient matrix of the above system of linear equations is a Vandermonde matrix in the variables $\{\Ga_{i,j}: i\in I,j\in S_i\}$. Since $\Ga_{i,j}$ are pairwise distinct, the determinant of the coefficient matrix is nonzero. Hence, such an system of linear equations has a unique solution $\Gl_{i,j}=0$ for $i\in I$ and $j\in S_i$.

{\bf Case 2:}  the last column of some blocks are contained in $\{\bc_{i,j}: 1\le i\le m+1,j\in S_i\}$.
We have already proved that any $d$ columns of $H$ are linearly independent for this case in Lemma \ref{lem:4.1}. So we omit the details here.
\end{proof}

\begin{prop}\label{prop:4.6}
Let $\mC$ be the code with $H$ given in Equation  \eqref{eq:lem4.5} as its parity-check matrix. Then $\mC$ is an $[n+\lceil \frac{n}{r}\rceil,n-d+1,\ge d+1;r]$-locally repairable code.
\end{prop}
\begin{proof}
The parity-check matrix $H$ is a $(\lceil \frac{n}{r}\rceil+d-1) \times (n+\lceil \frac{n}{r}\rceil)$ matrix and the rows of $H$ are linearly independent from linear algebra. Hence, the dimension of $\mC$ is $n-d+1$.
From Lemma \ref{lem:4.5} and \cite[Theorem 4.5.6]{LX04}, the minimum distance of $\mC$ is at least $d+1$.  Finally, the locality of $\mC$ is $r$ from the parity-check matrix $H$ and Corollary \ref{cor:2.7}.
\end{proof}

\begin{theorem}\label{thm:4.7}
Let $r$ be a positive integer. Let $n\le q-1$ be a positive integer such that $n=mr+s$ with $1\le s\le r$ and let $d\le s$ be a positive integer. Then there exists a Singleton-optimal $[n+\lceil \frac{n}{r}\rceil,n-d+1,d+1; r]_q$-locally repairable code.
\end{theorem}
\begin{proof}
From Proposition \ref{prop:4.6}, we have
\begin{align*}
n+\left\lceil \frac{n}{r} \right\rceil-(n-d+1)-\left\lceil \frac{n-d+1}{r} \right\rceil+2 &=d+1+\left\lceil \frac{n}{r} \right\rceil-\left\lceil \frac{n-d+1}{r} \right\rceil\\
&=d+1+\left\lceil \frac{mr+s}{r} \right\rceil-\left\lceil \frac{mr+s-d+1}{r} \right\rceil\\ &=d+1\le d(\mC).\\
\end{align*}
Hence, $\mC$ is a Singleton-optimal $[n+\lceil \frac{n}{r}\rceil,n-d+1,d+1;r]_q$-locally repairable code.
\end{proof}

\begin{ex}
{\rm Let $q=2^6=64$, $n=63$ and $r=11$. It is clear that $63=11\times 5+8$, i.e., $s=8$.
Choose $d=8$. From Theorem \ref{thm:4.7}, there exists a Singleton-optimal $[69,56,9;11]_{64}$-locally repairable codes with locality $r=11$.}
\end{ex}

\begin{rmk}
From Theorem \ref{thm:4.7}, we can explicitly construct a new family of Singleton-optimal locally repairable codes with length up to $q-1+\lceil\frac{q-1}{r}\rceil$, which is larger than $q$ compared with \cite{CXHF18,JMX20,TB14}.  
If the minimum distance $d$ is large and linearly proportional to $q$, then $r$ is linear proportional to $q-1$ as well, i.e., $\lceil\frac{q-1}{r}\rceil$ is a constant with respect to large $q$. 
In this case, the length of such Singleton-optimal locally repairable codes is less than the $q+2\sqrt{q}$ which can be obtained from elliptic curves \cite{LMX19,MX21}.
\end{rmk}

\subsection{Asymptotic bounds of locally repairable codes}
In \cite{MX20}, the authors gave a very technical and complicated method to obtain an asymptotic bound exceeding the Gilbert-Varshamov bound for locally repairable codes via local expansions of carefully chosen functions in the Garcia-Stichtenoth tower. In this subsection, we provide a much simpler proof for such an asymptotic bound given in \cite[Theorem 5]{MX20} via lengthening  algebraic geometry codes.

Let $r$ be a fixed positive integer.
Let $\{C_i\}_{i=1}^{\infty}$ be a family of $q$-ary $[n_i,k_i,d_i]$-linear codes with information rate and relative minimum distance $R_1=\lim_{i\rightarrow \infty} \frac{k_i}{n_i}$ and $\delta_1=\lim_{i\rightarrow \infty} \frac{d_i}{n_i}.$
From Proposition \ref{prop:4.2}, there exists a family of $[n_i+\lceil \frac{n_i}{r}\rceil, k_i,\ge d_i;r]$-locally repairable codes.
Then the information rate of this family of locally repairable codes is
$$R=\lim_{i\rightarrow \infty} \frac{k_i}{n_i+\lceil \frac{n_i}{r}\rceil}=\lim_{i\rightarrow \infty} \frac{k_i}{n_i} \times \frac{n_i}{n_i+\lceil \frac{n_i}{r}\rceil}= \frac{r}{r+1}R_1$$
and its relative minimum distance is $$\delta=\lim_{i\rightarrow \infty} \frac{d_i}{n_i+\lceil \frac{n_i}{r}\rceil}=\lim_{i\rightarrow \infty} \frac{d_i}{n_i}\times \frac{n_i}{n_i+\lceil \frac{n_i}{r}\rceil}= \frac{r}{r+1}\delta_1.$$

Combining with the classical Gilbert-Varshamov bound \cite[Proposition 8.4.4]{St09},
we have the following result.
\begin{prop}\label{prop:4.17}
Let $H_q(x)$ be the $q$-ary entropy function $H_q(x):=x\log_q(q-1)-x\log_q(x)-(1-x)\log_q(1-x).$ Then there exists a family of $q$-ary locally repairable codes with locality $r$ whose information rate $R$ and relative minimum distance
$\delta$ satisfy
$$R\ge \frac{r}{r+1}  \left(1-H_q\Big{(}\frac{r+1}{r}\delta\Big{)}\right).$$
\end{prop}

Combining with the TVZ bound of algebraic geometry codes \cite[Proposition 8.4.7]{St09}, we have 
$$\frac{r+1}{r} R\ge  1-\frac{1}{A(q)}-\frac{r+1}{r}\delta.$$
Hence, we have shown the following result which is the same as \cite[Theorem 5]{MX20}.
\begin{theorem}\label{thm:4.11}
Let $q$ be a prime power and let $A(q)$ be the Ihara's constant.
Then there exists a family of $q$-ary locally repairable codes with locality $r$ whose information rate $R$ and relative minimum distance
$\delta$ satisfy
\begin{equation}\label{eq:14}
R\ge \frac{r}{r+1}-\frac{1}{A(q)}\times \frac{r}{r+1}-\delta.
\end{equation}
\end{theorem}

\begin{figure}
\centering
\includegraphics[width=3.5in]{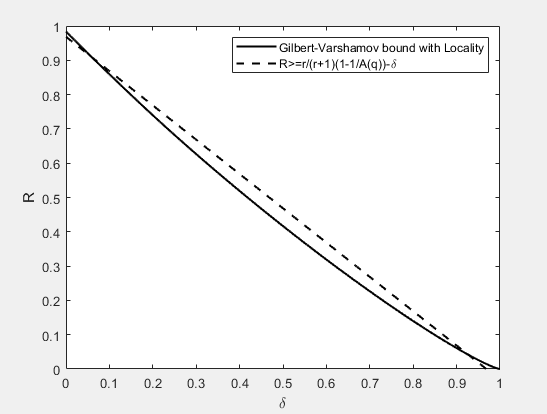}
\caption{$q=2^{12},r=61$}\label{Fig:4}
\end{figure}

The Figure \ref{Fig:4} shows that the bound given in Theorem \ref{thm:4.11} can exceed the asymptotic Gilbert-Varshamov bound of locally repairable codes for $q=2^{12}, r=61$.
Hence, we greatly simplify the bound given in \cite[Theorem 5]{MX20} that breaks the asymptotic Gilbert-Varshamov bound for locally repairable codes.

\end{document}